\documentclass[12pt,final,
thmsa,
epsf,a4paper]
{article}
\newtheorem{theorem}{Theorem}

\newtheorem{corollary}{Corollary}
\newtheorem{definition}{Definition}
\newtheorem{example}{Example}
\newtheorem{lemma}{Lemma}
\newtheorem{remark}{Remark}

\newenvironment{proof}[1][Proof]{\emph{#1.} }{\  \hfill $\square $ \vspace{5 pt}}

\usepackage{natbib}
\usepackage{amssymb}
\usepackage[dvips]{graphics}
\usepackage{amsmath}
\usepackage{enumerate}

\usepackage{mathpazo}

\usepackage{hyperref}
\hypersetup{
    colorlinks,
    citecolor=blue,
    filecolor=black,
    linkcolor=blue,
    urlcolor=blue
}

\usepackage{pgf,tikz}
\usetikzlibrary{arrows}

\usepackage[utf8]{inputenc}
\usepackage{amssymb, amsmath,geometry}
\usepackage{fancyhdr}

\usepackage{emptypage}
\usepackage{soul}

\usepackage{multirow}
\usepackage{array}
\usepackage{mathtools}
\usepackage{arydshln}
\usepackage{threeparttable}
\usepackage{booktabs,caption}

\begin{document}

\title{Variable population manipulations of 
reallocation rules in economies with single-peaked preferences\footnote{I thank Pablo Arribillaga, Jordi Mass\'o, William Thomson, the Associate Editor, and two anonymous referees for very detailed comments.  Of course, all errors are my own. I acknowledge financial support from Universidad Nacional de San Luis  through grants 032016 and 030320, from Consejo Nacional
de Investigaciones Cient\'{\i}ficas y T\'{e}cnicas (CONICET) through grant
PIP 112-200801-00655, and from Agencia Nacional de Promoción Cient\'ifica y Tecnológica through grant PICT 2017-2355. Part of this work was written during my visit at the Department of Economics, University of Rochester, under a Fulbright scholarship.}}

\author{Agustín G. Bonifacio\footnote{Instituto de Matemática Aplicada San Luis (UNSL-CONICET), Av. Italia 1556,  5700, San Luis, Argentina. E-mail: \href{mailto:agustinbonifacio@gmail.com}{ \texttt{abonifacio@unsl.edu.ar}}}}

\date{\today}

\maketitle

\begin{abstract}

In a one-commodity economy  with single-peaked preferences and individual endowments, we study different ways in which reallocation rules can be strategically distorted by affecting the set of active agents. We introduce and characterize the family of \emph{iterative} reallocation rules  and show that each rule in this class  is \emph{withdrawal-proof} and \emph{endowments-merging-proof}, at least one is \emph{endowments-splitting-proof} and that no such rule  is \emph{pre-delivery-proof.} 
\vspace{5 pt}

\noindent
\emph{Journal of Economic Literature}  Classification Numbers: D63, D71, D82 


\end{abstract}

\vspace{10 pt}

\section{Introduction}
This paper studies a variable population model of an economy consisting of one non-disposable commodity and a set of agents with  individual endowments of that commodity. In this context, we investigate different ways in which reallocation rules, i.e.,  systematic ways of selecting reallocations for each possible configuration of agents' preferences and endowments, can be strategically distorted by affecting the set of active agents. We limit our analysis  to the case where agents' preferences are single-peaked: up to some critical level, called the peak, an increase in an agent's consumption raises her welfare; beyond that level, the opposite holds. This model has been extensively studied \citep[see, for example,][]{Bonifacio15,Klaus97,Klaus98}. We allow for variable population as in \cite{Moreno96,Moreno02}.\footnote{One-commodity economies with single-peaked preferences and \emph{a social endowment} are studied, for example, by \cite{Sprumont91} assuming a fixed population. The first studies of the extension of this model to allow  for variable populations are \cite{thomson1994consistent, thomson1995population}.} To illustrate this type of problem, consider the distribution of a task (e.g., teaching hours) among the members of a group with concave disutility of labor (which induces single-peaked preferences over the time they dedicate to work). From one period to the next one, external factors (e.g., research and administrative duties) might affect preferences and a reallocation of the time assigned to each agent in the first period (taken as a benchmark for the second period) could benefit everyone. Another application of this model is a pollution problem in which countries have different rates of pollution and could trade via money transfers their pollution quotas.

Our analysis will be conducted over reallocation rules which are \emph{own-peak-only} (the sole information collected by the rule from an agent's preference to determine her reallocation is her peak amount) and meet the \emph{endowments lower bound} (no agent is made worse off than at her endowment). Two monotonicity properties are appealing in this model. First, 
a population monotonicity. Since variable population is allowed, it is natural to ask for  a monotonicity condition  requiring the arrival of new agents to affect all agents present before the change in the same direction. Adding the proviso that agents entering the economy do not change the sign of the excess demand of the economy, we get the property of \emph{one-sided population monotonicity}  \citep[see][]{thomson1995population}. Second, a resource monotonicity. If, in case of excess demand, the individual endowments decrease (or increase in case of excess supply), then no individual is better off after the change. We call this property \emph{one-sided endowments monotonicity} \citep[see][]{thomson1994resource}.

Our first result is a characterization of the family of reallocation rules that  satisfy the four previously mentioned properties (Theorem \ref{sequential}). This family of reallocation rules, which we call  ``iterative'', resembles the family of weakly sequential reallocation rules \citep[][]{Bonifacio15} in that each rule in the family can be described by a step-by-step procedure that, at each stage, guarantees levels of consumption to the agents which are adjusted throughout  the process. Weakly sequential reallocation rules, in turn, follow closely the definition of the sequential rules presented in \cite{Barbera97} for economies with single-peaked preferences and a social endowment to be allotted.\footnote{\cite{ehlers2002strategy} extends the rules of \cite{Barbera97} to the domain of single-plateaued preferences.} Our characterization turns out to be different from \cite{Barbera97}'s since their characterization of sequential rules marries \emph{efficiency} (allocations chosen by the rule satisfy the Pareto criterion) with \emph{strategy-proofness} (no agent ever has incentives to misrepresent preferences) and \emph{replacement monotonicity} (if a change in one agent's preferences results in that agent receiving at least as much as before, then all other agents receive at most as much as before). Iterative reallocation rules are always \emph{efficient} (Lemma \ref{lemaEMimpliesEFF}) but need not be \emph{strategy-proof} nor \emph{replacement monotonic}.\footnote{See Appendix \ref{segundo} for more details.}

Most of the literature on asymmetric rationing with single-peaked preferences  has been focused on a specific property of immunity to manipulation:  \emph{strategy-proofness}. \cite{moulin1999rationing} characterizes the family of rationing rules along fixed paths by means of \emph{efficiency}, \emph{strategy-proofness}, a \emph{resource monotonicity}, and \emph{consistency} \citep[see also][]{ehlers2002resource}.\footnote{See also \cite{moulin2017one}, and references therein, for a general treatment of strategy-proof allocation with single-peaked preferences.}  We deliberately depart from that stance. Instead, and following \cite{thomson2014new},  we examine the robustness of iterative reallocation rules to various types of manipulations by  affecting the set of active agents. The manipulations we consider are the following:

\begin{enumerate}[(i)]

\item Instead of participating, an agent withdraws with her endowment.  The rule is applied without her.  She then trades with one of the agents that did participate the resources they control between the two of them  in such a way  that both end up better off.  A rule immune to this type of manipulation is called \emph{withdrawal-proof}. 

\item Two agents can merge their endowments, and one of them withdraw. The rule is applied without this second agent. The agent who stays may be assigned an amount that can be divided between the two  in such a way that both become at least as well off as they would have been without the manipulation, and at least one of them is better off. A rule immune to this type of manipulation is called \emph{endowments-merging-proof}.

\item An agent may split her endowment with some outsider (an agent with no endowment). The rule is applied and the guest then transfers her assignment to the agent who invited her in.  The first agent may prefer her final assignment to what she would have received without the manipulation. A rule immune to this type of manipulation is called \emph{endowments-splitting-proof}. 

\item An agent may make a pre-delivery to some other agent of the trade the latter would  be assigned if she participated. The rule is applied without the second agent. At her final assignment, the first agent may be better off than she would have been without the manipulation. A rule immune to this type of manipulation is called \emph{pre-delivery-proof}.

\end{enumerate}

It turns out that all iterative rules satisfy both \emph{withdrawal-proofness} and \emph{endowments-merging-proofness} (Corollary \ref{cor} and  Lemma \ref{Th emp}). \emph{Endowments-splitting-proofness} is satisfied by the proportional reallocation rule (Remark \ref{propprop}). Finally,  no iterative reallocation rule satisfies \emph{pre-delivery-proofness} (Corollary \ref{corpre}).

 The rest of the paper is organized as follows. In Section \ref{preliminar} the model and some basic properties of reallocation rules are presented. In Section \ref{sectionsequential}, iterative reallocation rules are defined and characterized. The different variable population manipulations are discussed in Section \ref{sectionvariable}. Final comments are gathered in Section \ref{sectionfinal}.

\section{Preliminaries}\label{preliminar}

\subsection{Model}
 We consider the set of natural numbers $\mathbb{N}$ as the set of  \textbf{potential agents.} Denote by $\mathcal{N}$ the collection of all finite subsets of   $\mathbb{N}.$ Each $i \in \mathbb{N}$ is characterized by an endowment 
 $\omega_i \in \mathbb{R}_+$ of the good and a continuous preference relation $R_i$ defined over
$\mathbb{R}_+$. Call $P_i$ and $I_i$ to the strict preference and indifference relations associated with $R_i,$ respectively. We assume that agents' preferences are \textbf{single-peaked}, i.e.,
each $R_i$ has a unique maximum $p(R_i) \in \mathbb{R}_+$  such that, for each pair $\{x_i, x_i'\} \subset\mathbb{R}$, we have $x_iP_ix_i'$ as long as either $x_i'<x_i\leq p(R_i)$ or $p(R_i) \leq x_i<x_i'$ holds. Denote by $\mathcal{R}$ the domain of
single-peaked preferences defined on $\mathbb{R}_+$. Given $N \in \mathcal{N},$ 
an  \textbf{economy} 
consists of a profile of preferences $R \in \mathcal{R}^N$ and an
individual endowments 
vector $\omega = (\omega_i)_{i \in N}\in \mathbb{R}^N_{+}$, and is
denoted by $e=(R,\omega)$. If $S \subset N$ and $R \in
\mathcal{R}^N,$ let $R_S=(R_j)_{j \in S}$ denote the restriction of
$R$ to $S.$ We often write $N\setminus S$ by $-S$.  Using similar notation for the vector of endowments,  $e'=(R_S',R_{-S}, \omega_S',\omega_{-S})$ stands for the
economy where the preference and endowment of agent $i \in S$ are $R_i'$ and
$\omega_i',$  and those of agent $i \notin S$ are $R_i$
and $\omega_i.$  Let 
$\mathcal{E}^N$ be the domain of economies with agents in $N$. Given $e=(R,\omega) \in \mathcal{E}^N,$  let $z(e)=\sum_{j \in N}(p(R_j)-\omega_j).$ If $z(e) \geq 0$ we say that economy $e$ has \textbf{excess demand} whereas if $z(e)<0$ we say that economy $e$ has \textbf{excess supply}. Let  $\mathcal{E}=\bigcup_{N \in \mathcal{N}} \mathcal{E}^{N}$ denote the set of all potential economies. For each $N \in \mathcal{N}$ and each $e \in \mathcal{E}^N,$ let $X(e)=\{x \in  \mathbb{R}^{N}_+ : \sum_{j\in N}x_j=\sum_{j \in N}\omega_j\}$ be the set of \textbf{reallocations} for economy $e,$ and let $X=\bigcup_{e \in \mathcal{E}}X(e).$ A  \textbf{reallocation rule} is a function $\varphi: \mathcal{E} \to X$ such that $\varphi(e) \in X(e)$ for each $e \in \mathcal{E}.$  For each $N \in \mathcal{N},$ each $i\in N,$ and each $e \in \mathcal{E}^N,$ let $\Delta \varphi_i(e)=\varphi_i(e)-\omega_i$ be \textbf{agent $\boldsymbol{i$'s net trade at $e}$}.

\subsection{Basic properties}\label{basic properties}

The next informational simplicity property states that if an  agent unilaterally changes   her preference for another one with the same peak, then  her allotment remains unchanged.   

\vspace{5 pt}

\noindent
\textbf{Own-peak-only:} For each $e=(R, \omega) \in \mathcal{E}^N,$ each $i \in N,$ and each $ R_i' \in  \mathcal{R}$ such that $p(R_i')=p(R_i),$ if $e'=(R_i', R_{N\setminus \{i\}}, \omega)$ then $\varphi_i(e')=\varphi_i(e).$  
\vspace{5 pt}

\noindent This property is weaker than the ``peak-only'' property,\footnote{See Appendix \ref{segundo} for a discussion on this fact.} that has been imposed in a number of axiomatic studies. Analyzing the uniform rule, \cite{Sprumont91} derives the own-peak-only property from other axioms \citep[see also][]{Ching92,Ching94}.

The following property requires respecting ownership of the resource, and also can be seen as giving incentive to participate in the exchange process. It says that no agent can get a reallocation that she finds worse than her endowment.

\vspace{5 pt}
\noindent
\textbf{Endowments lower bound:}  For each $e=(R, \omega) \in \mathcal{E}^N,$ and each $i \in N,$ $\varphi_i(e)R_i \omega_i.$\footnote{This property is commonly known in the literature as \emph{individual rationality.}} 
\vspace{5 pt}

Next, we present our two monotonicity properties. The first one requires that as population enlarges, and the new resources and preferences considered are not as disruptive as to modify the status of the economy from excess demand to excess supply or vice versa,  the welfare of each of the initially present agents should move in the same direction.

\vspace{5 pt}
\noindent
\textbf{One-sided population monotonicity:} For each $e=(R, \omega) \in \mathcal{E}^N,$ each $N' \subset N,$ and each $e'=(R_{N'}, \omega_{N'}) \in \mathcal{E}^{N'},$  $z(e) \ z(e') \geq 0$ implies either $\varphi_i(e) R_i \varphi_i(e')$ for each $i \in N'$ or $\varphi_i(e')R_i \varphi_i(e)$ for each $i \in N'$.
\vspace{5 pt}

\noindent The second one requires all agents to benefit from a favorable change in the amount to allocate. 
Given two vectors $x, y \in \mathbb{R}^N,$ define $x \geqslant y$ if and only if $x_i \geq y_i$ for each $i \in N.$  

\vspace{5 pt}
\noindent
\textbf{One-sided endowments monotonicity:} For each $e=(R, \omega) \in \mathcal{E}^N,$ and each $\omega' \in \mathbb{R}^N$ such that $\omega' \geqslant \omega,$  if $e'=(R, \omega'),$ then $z(e') \geq 0$ implies $\varphi_i(e')R_i\varphi_i(e)$ for each $i \in N,$ and $z(e) \leq 0$ implies $\varphi_i(e)R_i\varphi_i(e')$ for each $i \in N$.
\vspace{5 pt}

The usual Pareto optimality property states that, for each economy, the reallocation selected by the rule should be such that there is no other reallocation that all agents find at least as desirable and at least one agent prefers. In this model, it is equivalent to the following same-sidedness condition:

\vspace{5 pt}
\noindent
\textbf{Efficiency:} For each $e=(R, \omega) \in \mathcal{E}^N, $ $z(e) \geq 0$ implies  $\varphi_i(e)\leq p(R_i)$ for each $i \in N,$
and  $z(e) \leq 0$ implies  $\varphi_i(e)\geq p(R_i)$ for each $i \in N.$
\vspace{5 pt}

\begin{lemma}\label{lemaEMimpliesEFF}
Each own-peak-only and one-sided endowments monotonic reallocation rule that  meets the endowments lower bound  is efficient. 
\end{lemma}
\begin{proof}
Let $\varphi$ be an \emph{own-peak-only, one-sided endowments monotonic} rule that meets the \emph{endowments lower bound}, and assume $\varphi$ is not \emph{efficient}. Then, there is $e=(R, \omega) \in \mathcal{E}^N$, 
  that without loss of generality we assume $z(e)\geq 0$, and 
 $i \in N$ such that $\varphi_i(e)>p(R_i).$ This implies that 
\begin{equation}\label{lemEM1}
\varphi_i(e) \leq \omega_i.    
\end{equation}
 To see that \eqref{lemEM1} holds, first assume $p(R_i)\leq \omega_i < \varphi_i(e)$. By single-peakedness, $\omega_iP_i\varphi_i(e)$, contradicting the \emph{endowments lower bound}. Second, assume $\omega_i<p(R_i)$. Let $\widetilde{R}_i \in \mathcal{R}$ be such that $p(\widetilde{R}_i)=p(R_i)$, and $\omega_i\widetilde{P}_i\varphi_i(e)$ and let $\widetilde{e}=(\widetilde{R}_i, R_{-i}, \omega)$. By the \emph{own-peak-only} property, $\varphi_i(\widetilde{e})=\varphi_i(e).$ Hence, $\omega_i\widetilde{P}_i\varphi_i(\widetilde{e})$, contradicting the \emph{endowments lower bound}. Therefore, \eqref{lemEM1} holds. Next, let $\omega_i' \in \mathbb{R}_{+}$ be such that $\omega_i'=p(R_i)$, and let $e'=(R,\omega_i', \omega_{-i}).$ By the \emph{endowments lower bound}, $\varphi_i(e')=p(R_i)$ and therefore  $\varphi_i(e')P_i\varphi_i(e)$, contradicting \emph{one-sided endowments monotonicity} since $z(e')\geq 0$.  
\end{proof}

The following result  will be  useful in the rest of the paper.

\begin{lemma}\label{elb}
Let $\varphi$ be an efficient and own-peak-only  reallocation rule that meets the endowments lower bound. Let $e=(R, \omega) \in \mathcal{E}^N$ and $i \in N.$ If either $z(e) \geq 0 \text{ and } p(R_i) \leq \omega_i $, or $z(e) \leq 0\text{ and } p(R_i) \geq \omega_i$,  then $\varphi_i(e)=p(R_i).$ 
\end{lemma}
\begin{proof}
Let $\varphi$ satisfy the properties in the lemma and let $e \in \mathcal{E}^N$ and $i \in N$ . Assume $z(e) \geq 0$ and $p(R_i) \leq \omega_i.$ Since $z(e)\geq 0,$ by \emph{efficiency,} $\varphi_i(e) \leq p(R_i).$ If $p(R_i)=\omega_i$,  $\varphi_i(e)=p(R_i)$ by the \emph{endowments lower bound}.  Suppose, then, that $p(R_i)<\omega_i$ and  $\varphi_i(e) < p(R_i).$ Let $R_i' \in \mathcal{R}$ and $x_i \in \mathbb{R}_+$ be such that $p(R_i')=p(R_i),$ $\varphi_i(e)<x_i<p(R_i)$ and $x_i I_i' \omega_i.$  Let $e'=(R_i', R_{N\setminus\{i\}}, \omega).$ By the \emph{own-peak-only} property, $\varphi_i(e')=\varphi_i(e).$ Then, $\omega_i P_i' \varphi_i(e'),$ contradicting the \emph{endowments lower bound.} A similar reasoning establishes the same conclusion when $z(e) \leq 0\text{ and } p(R_i) \geq \omega_i.$ 
\end{proof}

\section{Iterative reallocation rules}\label{sectionsequential}

In this section we present a well-behaved family of reallocation rules. They resemble the weakly sequential reallocation rules\footnote{These rules, in turn, follow closely the definition of the sequential rules presented in \cite{Barbera97} for economies with a social endowment to be allotted.}   \citep[][]{Bonifacio15}  in that both families of rules are defined through an easy step-by-step procedure. The two families, although overlap, are different. This is discussed thoroughly in Appendix \ref{segundo}.

\subsection{Definition}

For each $N \in \mathcal{N}$ and  each $e=(R, \omega)\in \mathcal{E}^N,$ let $Q(e)\equiv \{q \in \mathbb{R}^N : \sum_{j \in N}q_j=0 \text{ and }\omega+q \geqslant 0\}$ be the possible net trades of endowments in economy $e$ and let $Q\equiv \bigcup_{e \in \mathcal{E}}Q(e).$ Next, define $\mathcal{Q}\equiv \{(q,e)\in Q\times \mathcal{E} : q \in Q(e)\}.$ Each element of $\mathcal{Q}$ specifies a net trade in a particular economy. An iterative reallocator is a function that, for each $N \in \mathcal{N}$ and  each economy $e=(R,\omega)\in \mathcal{E}^N$,  starting from the individual endowments of the agents (i.e., from a net trade $q_i^0$ equal to zero for each agent $i \in N$), generates  iteratively  a sequence of net trades $q^0, q^1, \ldots, q^{|N|-1}.$ Its repeated application is constrained to follow monotonic features.

\begin{definition}\label{reallocator} An \textbf{iterative reallocator} is a function $g: \mathcal{Q} \to \mathcal{Q}$ such that $g(q,e) \in \mathcal{Q}(e)$ and, for each $N \in \mathcal{N},$ and each $e=(R, \omega)\in \mathcal{E}^N,$  if $(q^t, e)=g(q^{t-1}, e)\equiv g^{t}(0,e),$  then:\footnote{Here $g^{t}$ denotes $g$ compose with itself $t$ times.} 
\begin{enumerate}[(i)]

\item for each $i \in N,$ and each $t\geq 1,$ 
$$q^t_i=p(R_i)-\omega_i\text{ \   whenever    \    } \begin{cases}
z(e)\geq 0 \text{ \ and \ } p(R_i)\leq \omega_i+q_i^{t-1}\\
z(e) <0 \text{ \ and  \ } p(R_i)\geq \omega_i+q_i^{t-1}.\\
\end{cases}
$$

\item for each $i \in N,$ and each $t\geq 1,$ $$q^t_i \geq q^{t-1}_i  \text{ \   whenever \  }z(e) \geq 0 \text{ and }p(R_i)>\omega_i+q_i^{t-1}$$
$$q^t_i \leq q^{t-1}_i  \text{ \   whenever \  }z(e) < 0 \text{ and }p(R_i)< \omega_i+q_i^{t-1}.$$

\item   for each $\widetilde{e}=(R,\widetilde{\omega})\in \mathcal{E}^N$ such that $\widetilde{\omega} \geqslant \omega$ and $(\widetilde{q}^{|N|-1},\widetilde{e})=g^{|N|-1}(0,\widetilde{e}),$    $$\widetilde{\omega}+\widetilde{q}^{|N|-1}\geq \omega+q^{|N|-1} \text{ \ whenever \ } z(e) \leq 0 \text{ \ or \ } z(\widetilde{e}) \geq 0.$$ 

\item for each $\widetilde{N} \subset N,$  each $\widetilde{e}=(R_{\widetilde{N}}, \omega_{\widetilde{N}}) \in \mathcal{E}^{\widetilde{N}},$  and $(\widetilde{q}^{|\widetilde{N}|-1},\widetilde{e})=g^{|\widetilde{N}|-1}(0,\widetilde{e}),$ $$ \left[\widetilde{q}_i^{{|\widetilde{N}|-1}}-q_i^{{|N|-1}}\right]\left[\widetilde{q}_j^{{|\widetilde{N}|-1}}-q_j^{|N|-1}\right]\geq 0 \text{ \ whenever \ }\begin{cases}
z(e)\geq 0 \text{ \ and \ } z(\widetilde{e}) \geq 0\\
z(e)\leq 0 \text{ \ and \ } z(\widetilde{e}) \leq 0\\
\end{cases}$$ 
for each $\{i,j\} \subset \widetilde{N}.$
\end{enumerate}

\end{definition}

Let us put in words the above definition for the  case of excess demand (this is, when $z(e)\geq 0$). The first two conditions relate to the behavior of the net trades of an economy throughout the iterations  of  $g.$ Condition (i) says that if at  stage $t-1$ agent $i$'s peak is not higher  than her endowment  \emph{plus} her net trade, i.e. $p(R_i) \leq \omega_i+q_i^{t-1},$  then agent $i$'s net trade is set at $p(R_i)-\omega_i$ from stage $t$ onward. Condition (ii) establishes that if at stage $t-1$ agent $i$'s peak is  higher  than her endowment  \emph{plus} her net trade, i.e. $p(R_i) > \omega_i+q_i^{t-1},$  then her net trade should not decrease from stage $t-1$ to stage $t,$ i.e. $q_i^{t}\geq q_i^{t-1}.$ The last two conditions relate to the behavior of the iterations of $g$ between two different economies. Condition (iii) states that, in another economy $\widetilde{e}$ with the same agents and preferences where no agent has lower endowment and the increase in the resources is not disruptive, i.e. $z(\widetilde{e}) \geq 0,$  the resources available to each agent in the last stage of the iterations cannot be smaller than the resources available to each  agent  in the last stage of the iterations in the original economy.   Finally, Condition (iv) says that for any subeconomy with excess demand, if the net trade of one agent in the last stage of the iteration is not  smaller (bigger) than the net trade that same  agent gets in the original economy, then the net trade of each of the other agents in the subeconomy  should not be smaller (bigger) than the net trade that  agent gets in the original economy either.

\begin{remark}\label{remseq}\emph{
Note that, as there are $|N|$ agents in the economy, at most $|N| -1$ adjustments take place. Therefore,  $(q^{|N|+t-1},e)=g^{|N|+t-1}(0,e)$ implies $q^{|N|+t-1}=q^{|N|-1}$ for each $t \geq 1.$ }
\end{remark}

Each iterative reallocator induces a reallocation rule in a straightforward way:

\begin{definition}
A reallocation rule $\varphi: \mathcal{E}\to X$ is \textbf{iterative} if there is an iterative reallocator $g:\mathcal{Q} \to \mathcal{Q}$ such that, for each $N \in \mathcal{N}$ and each $e \in \mathcal{E}^N,$  $(q^{|N|-1}, e)=g^{|N|-1}(0,e)$  implies  $\Delta \varphi(e)=q^{|N|-1}.$ 
\end{definition}

A prominent member of the class of iterative reallocation rules is the uniform reallocation rule, first studied by \cite{thomson1995axiomatic} \citep[see also][]{Klaus97, Klaus98}, that adapts the celebrated uniform rule characterized by \cite{Sprumont91} to the model with individual endowments:

\vspace{10 pt}


\noindent \textbf{Uniform reallocation rule, $\boldsymbol{u}$:} for each $N \in \mathcal{N},$ each $e\in \mathcal{E}^N,$ and each $i \in N,$
$$u_i(e)=\left\{\begin{array}{l l }
\min\{p(R_i), \omega_i+\lambda(e)\} & \text{if } z(e)\geq 0\\
\max\{p(R_i), \omega_i-\lambda(e)\} & \text{if } z(e)< 0\\
\end{array}\right.$$
where $\lambda(e)\geq 0$ and solves $\sum_{j \in N}u_j(e)=\sum_{j\in N}\omega_j.$ 

\vspace{10 pt}

\noindent Within the class of iterative reallocation rules, the uniform reallocation rule is the only one that supports \emph{envy-free} redistributions, meaning by this that for no $N \in \mathcal{N},$  $e\in \mathcal{E}^N,$ and pair of agents $\{i,j\} \subset N$ such that $\omega_i - \Delta u_j(e) \in \mathbb{R}_+,$ we have $(\omega_i - \Delta u_j(e) )P_i u_i(e)$ \citep[see][Theorem 1]{Moreno02}.

To see that the uniform reallocation rule is an iterative reallocation rule, given $e=(R,\omega)\in \mathcal{E}^N$ and $q^0=(0,\ldots, 0),$ 
 consider the iterative reallocator $g:\mathcal{Q} \to \mathcal{Q}$ defined as follows. If $(q^t,e)=g(q^{t-1},e)$ then, for each $i \in N$ and each $t=1, \ldots, |N|-1,$

$$q_i^t=\left\{\begin{array}{l l }
p(R_i)-\omega_i & \text{if  } i \in N^t\\
\lambda^{t} & \text{if  } i \in N\setminus N^t\\
\end{array}\right.$$
where 
 $$\lambda^t=\lambda^{t-1}+ \frac{\sum_{j \in N^t \setminus (\cup_{s=0}^{t-1}N^s)} [\omega_j+\lambda^{t-1}-p(R_j)]}{|N \setminus N^t|},$$
$\lambda^{0}=0$, $N^0=\emptyset$, and 
$$N^t=\left\{\begin{array}{l l }
\{j \in N : p(R_j) \leq \omega_j+q_j^{t-1}\} & \text{if } z(e)\geq 0\\
\{j \in N : p(R_j) > \omega_j +q_j^{t-1}\} & \text{if } z(e) < 0\\
\end{array}\right.$$

It is easy to see that $(q^{|N|-1},e)=g^{|N|-1}(0,e)$ implies $q^{|N|-1}=\Delta u(e)$ for each $N \in \mathcal{N}$ and each $e \in \mathcal{E}^N.$

\begin{example} \rm
Consider $e=(R,\omega)^{\{1,2,3,4\}}$ with $p(R_1)=0,$ $p(R_2)=2,$ $p(R_3)=3.5,$ and $p(R_4)=10$; and $\omega_1=9,$ $\omega_2=1,$ $\omega_3=0,$ and $\omega_4=2.$ Then, as $z(e)=15.5-12>0,$
\begin{description}
    \item \hspace{40 pt} $N^1=\{1\}, \ \lambda^1=0+\frac{9+0-0}{3}=3, \text{ \ and thus \  }q^1=(-9,3,3,3)$,
    \item \hspace{40 pt} $N^2=\{1,2\}, \ \ \lambda^2=3+\frac{1+3-2}{2}=4, \text{ \ and thus \  }q^2=(-9,1,4,4),$
    \item \hspace{40 pt} $N^3=\{1,2,3\},  \ \ \lambda^3=4+\frac{0+4-3.5}{1}=4.5, \text{ \ and thus \ }q^3=(-9,1,3.5,4.5)$.
\end{description} Therefore, $u_1(e)=0$,  
 $u_2(e)=2$, $u_3(e)=3.5$, and $u_4(e)=6.5.$  \hfill $\Diamond$
\end{example}

Another iterative reallocation rule, that will be analyzed in Section \ref{splitting}, is the proportional reallocation rule that we present next. 

\vspace{10 pt}

\noindent \textbf{Proportional reallocation rule, $\boldsymbol{\varphi^p}$:} for each $N \in \mathcal{N},$ each $e\in \mathcal{E}^N,$ and each $i \in N,$

$$\varphi^p_i(e)=\left\{\begin{array}{l l }
\min\{p(R_i), \lambda(e)\omega_i\} & \text{if } z(e)\geq 0\\
\max\{p(R_i), \lambda(e)\omega_i\} & \text{if } z(e)\leq 0\\
\end{array}\right.$$
where $\lambda(e)\geq 1$ and solves $\sum_{j \in N}\varphi_j^p (e)=\sum_{j\in N}\omega_j.$\footnote{For this rule to be well-defined, we need to constrain the domain of economies to those in which all individual endowments are always strictly positive. }


\subsection{Characterization}

The next result states that the class of iterative rules is characterized by the \emph{own-peak-only} property, the \emph{endowments lower bound}, \emph{one-sided endowments monotonicity} and \emph{one-sided population monotonicity:}

\begin{theorem}\label{sequential}
A reallocation rule satisfies the own-peak-only property, the endowments lower bound,  one-sided endowments monotonicity, and one-sided population monotonicity if and only if it is an iterative reallocation rule. 
\end{theorem}
\begin{proof} ($\Longrightarrow$) Let $\varphi$ be an \emph{own-peak-only}, \emph{one-sided endowments monotonic}, and \emph{one-sided population monotonic} reallocation rule that meets the  \emph{endowments lower bound}. 
 By Lemma \ref{lemaEMimpliesEFF}, $\varphi$ is also \emph{efficient}. The iterative reallocator  $g:\mathcal{Q} \to \mathcal{Q}$ is constructed as follows. Given $(q^{t-1},e)\in \mathcal{Q},$ define $q^t$ such that $(q^t, e)=g(q^{t-1},e)$ as 
\begin{equation}\label{master}
    q^t=\varphi(e^t)-\omega
\end{equation}
where  economy $e^t=(R, \omega^t) \in \mathcal{E}^N$ is such that, for each $i \in N,$  
$$\omega_i^t=\left\{
\begin{array}{l l}
p(R_i) & \text{if }z(e)[p(R_i)-\omega_i^{t-1}-q_i^{t-1}] \leq 0\\
\omega_i^{t-1}+q_i^{t-1} & \text{otherwise,}
\end{array}\right.$$ with $\omega_i^0=\omega_i$ and $q_i^0=0.$

Let us assume that $e=(R, \omega)\in \mathcal{E}^N$ is such that $z(e)\geq 0.$ The other case is similar. We need to see that $\Delta \varphi(e)=q^{|N|-1}$ where $q^{|N|-1}$ is such that $(q^{|N|-1},e)=g^{|N|-1}(0,e)$. 
In order to do this, 
for each  $t=1, \ldots, |N|-1,$ let $q^t$ be such that $(q^t,e)=g(q^{t-1},e)=g^t(0,e)$   (notice that $q^0=0$). 

\noindent \textbf{Claim 1: Let  $\boldsymbol{t \in \{1, \ldots, |N|-1\}}.$ If $\boldsymbol{p(R_i)\geq \omega_i^{t-1}+q_i^{t-1}}$ for each $\boldsymbol{i \in N,}$} \textbf{then} $\boldsymbol{q^t=\Delta \varphi(e)}$. Consider first the case $t=1.$ As $q^0=0,$ by the hypothesis $p(R_i) \geq \omega_i$ for each $i \in N.$ The \emph{endowments lower bound, efficiency,} and feasibility imply $\varphi(e)=\omega,$ and therefore $\Delta \varphi(e)=0.$  Note that, since in $e^1$ no agent's peak is less than her endowment and $z(e^1)\geq 0$, by the same reasoning as before $ \varphi(e^1)=\omega.$ Thus, $q^1=\omega-\omega=0.$ Next, assume the claim is true for each $t< T.$   Then $q_i^{T-1}=0$ for each $i \in N$ and, again, since  in $e^T$ no agent's peak is less than her endowment and $z(e^T)\geq 0,$ we get $q_i^T=0=\Delta \varphi(e).$ This proves the claim. 

\noindent \textbf{Claim 2: Let $\boldsymbol{t \in \{1, \ldots, |N|-1\}}.$ If $\boldsymbol{i \in N}$ is such that $\boldsymbol{p(R_i) \leq  \omega_i^{t-1}+q_i^{t-1}},$  then $\boldsymbol{q^t_i=\Delta \varphi_i(e)}$}. Let $i \in N$ be such  that $p(R_i) \leq  \omega_i^{t-1}+q_i^{t-1}.$ First, notice that when $t=1$,  $p(R_i) \leq \omega_i^0+q_i^{0}$ implies, as $q^0_i=0$ and $\omega_i^0=\omega_i$, that   $p(R_i) \leq \omega_i.$ Then,  by Lemma \ref{elb}, $\varphi_i(e)=p(R_i)$ and, therefore,
\begin{equation}\label{equ1}
    \Delta \varphi_i(e)=p(R_i)-\omega_i.
\end{equation}
Next, let $t \in \{1, \ldots, |N|\}$. Since $\omega_i^t=p(R_i)$ and  $z(e^t)\geq 0,$ by Lemma \ref{elb} applied to economy $e^t$, we have $\varphi_i(e^t)=p(R_i)$ and then $q^t_i=p(R_i)-\omega_i.$ Hence, by \eqref{equ1},  $q^t_i=\Delta \varphi_i(e).$ This proves the claim. 

Claims 1 and 2 show that if $q^{|N|-1}$ is such that  $(q^{|N|-1},e)=g^{|N|-1}(0,e)$, then $q^{|N|-1}=\Delta \varphi(e).$ It remains to be checked that function $g$ satisfies conditions (i)-(iv) in Definition \ref{reallocator}. Condition (i) is clear by Claims 1 and 2. Condition (ii) follows from the next claim.

\noindent \textbf{Claim 3: for each $\boldsymbol{t=1, \ldots, |N|-1,}$ if $\boldsymbol{i \in N}$ is such that $\boldsymbol{p(R_i) > \omega_i^{t-1}+q_i^{t-1}},$  then $\boldsymbol{q^t_i\geq q_i^{t-1}}$}. Let $i \in N$ be such that $p(R_i) > \omega_i^{t-1}+q_i^{t-1}$. Consider first the case $t=1.$ Since $q^0=0$ and $\omega_i^0=\omega_i,$ by the hypothesis $p(R_i)>\omega_i.$ Then, $\omega_i^1=\omega_i.$ This implies, as $z(e^1)\geq 0,$ that $\varphi_i(e^1)\geq \omega_i^1$ by the \emph{endowments lower bound.} Hence, $q_i^1=\varphi_i(e^1)-\omega_i\geq 0=q^0_i$ and thus $q_i^1 \geq q^0_i$. Next, assume the claim is true for each $t< T.$ Then $q_i^{T-1}\geq q_i^{T-2}\geq \ldots \geq q^0_i=0.$ Since $p(R_i)>\omega_i^{T-1}+q_i^{T-1}$ implies $\omega_i^T=\omega_i^{T-1}+q_i^{T-1}$ and $z(e^T)\geq 0,$ by the \emph{endowments lower bound} $\varphi_i(e^T)\geq \omega_i^T=\omega_i+\sum_{k=1}^{T-1}q_i^k.$ Then, $$q_i^T=\varphi_i(e^T)-\omega_i\geq \sum_{k=1}^{T-1}q_i^k\geq q_i^{T-1}.$$ This proves the claim.

Condition (iii) follows from the definition of $g$ and \emph{one-sided endowments monotonicity} of $\varphi,$ whereas condition (iv) is a consequence of the definition of $g$ and \emph{one-sided population monotonicity} of $\varphi.$

\noindent ($\Longleftarrow$) Let $\varphi$ be an iterative reallocation rule. Then there exists an iterative reallocator $g$ such that, for each $e=(R,\omega) \in \mathcal{E}^N,$ if $q^{|N|-1}$ is such that $(q^{|N|-1},e)=g^{|N|-1}(0,e)$, then $\Delta \varphi(e)=q^{|N|-1}.$ We will consider only the case $z(e)\geq 0,$ since an analogous argument can be used in the case $z(e)<0$. Next, we prove that $\varphi$ is \emph{efficient},\footnote{This is used to prove the two monotonicity properties.} \emph{one-sided endowments monotonic}, and \emph{one-sided population monotonic}.
\vspace{5 pt}
\\
\textbf{Efficiency}: We need to show that $\varphi_i(e)\leq p(R_i)$ for each $i \in N.$ Suppose $\varphi_i(e)\neq p(R_i).$ Then $\omega_i+q_i^{|N|-1}\neq p(R_i)$. If $q_i^{|N|-1} >p(R_i)-\omega_i,$ by  (i) in Definition \ref{reallocator} we have $q_i^{|N|}=p(R_i)-\omega_i.$ Then, $q_i^{|N|-1}\neq q_i^{|N|},$ contradicting Remark \ref{remseq}. Thus, $q_i^{|N|-1} \leq p(R_i)-\omega_i$ implying $\Delta \varphi_i(e)=q_i^{|N|-1} \leq p(R_i)-\omega_i$ and $\varphi_i(e) \leq p(R_i).$ 

 \vspace{5 pt}

\noindent \textbf{One-sided endowments monotonicity}: Let $\widetilde{e}=(R, \widetilde{\omega}) \in \mathcal{E}^N$ be such that $\widetilde{\omega}\geqslant \omega$ and $z(\widetilde{e})\geq 0,$ let $\widetilde{q}^{|N|-1}$ be such that  $(\widetilde{q}^{|N|-1},\widetilde{e})=g^{|N|-1}(0,\widetilde{e})$ and consider $i \in N$ such that $p(R_i) > \widetilde{\omega}_i.$ By condition (iii) in Definition \ref{reallocator}, $\widetilde{\omega}_i+\widetilde{q}_i^{|N|-1}\geq \omega_i+q_i^{|N|-1}.$ Then, $\varphi_i(\widetilde{e})\geq \varphi_i(e)$ and, as by \emph{efficiency} $p(R_i)\geq \varphi_i(\widetilde{e}),$ we have $\varphi_i(\widetilde{e})R_i\varphi_i(e).$ 
\vspace{5 pt}
\\
\textbf{One-sided population monotonicity}: Let $\widetilde{N} \subset N$ and   $\widetilde{e}=(R_{\widetilde{N}}, \omega_{\widetilde{N}}) \in \mathcal{E}^{\widetilde{N}}$ be such that $z(\widetilde{e})\geq 0.$  Take $\{i,j\} \subseteq \widetilde{N}.$ By condition (iv) in Definition \ref{remseq},  $[\widetilde{q}_i^{{|\widetilde{N}|}}-q_i^{{|N|}}][\widetilde{q}_j^{{|\widetilde{N}|}}-q_j^{|N|}]\geq 0.$  Assume, without loss of generality, that $\widetilde{q}_i^{{|\widetilde{N}|}}\geq q_i^{{|N|}}.$ Then, $\widetilde{q}_j^{{|\widetilde{N}|}}\geq q_j^{{|N|}}.$ This implies $\varphi_i(\widetilde{e})\geq \varphi_i(e)$ and $\varphi_j(\widetilde{e})\geq \varphi_j(e).$ As $z(\widetilde{e})\geq 0,$ by \emph{efficiency,} $p(R_i)\geq \varphi_i(\widetilde{e})$ and $p(R_j)\geq \varphi_j(\widetilde{e}).$ Thus, $\varphi_i(\widetilde{e})R_i\varphi_i(e)$ and $\varphi_j(\widetilde{e})R_i\varphi_j(e).$
\vspace{5 pt}
\\
To complete the proof, notice that $\varphi$ satisfies the  \emph{own-peak-only} property because $g$ does, and meets the \emph{endowments lower bound} because, for each agent, the adjustment process at each step guarantees an amount at least as good as the individual endowment. 
\end{proof}

\vspace{5 pt}

The independence of the axioms involved in the characterization of Theorem  \ref{sequential} is analyzed in Appendix \ref{appendix}. 

\section{Variable population manipulations}\label{sectionvariable}

In this section, we analyze each of the four properties of immunity to manipulation presented in the introduction and its relations with the family of iterative reallocation rules.

\subsection{Withdrawal-proofness}\label{sectionwithdrawal}

Consider an economy and suppose that an agent withdraws with her endowment and the reallocation  rule is applied without her. It could be the case that the amount that some other agent received   in the reallocation together with the endowment of the agent that withdrew could be re-divided between the two  of them  in such a way that both agents get (strictly) better off with respect to the assignments they would have obtained if the first agent had not withdrawn. We require immunity to this sort of behavior:

\vspace{5 pt}
\noindent
\textbf{Withdrawal-proofness:} For each $e=(R,\omega) \in
\mathcal{E}^N,$ each $\{i, j\} \subset N$ and each 
$(x_i,x_j) \in \mathbb{R}^2_+$ such that
$x_i+ x_j=\varphi_i(e')+\omega_j,$ where $e' = (R_{N\setminus\{j\}}, \omega_{N\setminus\{j\}}),$ it is not the case that $x_k P_k \varphi_k(e)$ for each
$k \in \{i,j\}.$
\vspace{5 pt}

 The next result shows that \emph{withdrawal-proofness} is implied by some of the properties discussed in Subsection \ref{basic properties}. 

\begin{lemma}\label{Th wp}
Each efficient, own-peak-only, and one-sided population monotonic  reallocation rule that meets the endowments lower bound is withdrawal-proof.  
\end{lemma}
\begin{proof}
Let $\varphi$ satisfy the hypothesis of the Theorem. By Lemma \ref{lemaEMimpliesEFF}, $\varphi$ is also \emph{efficient}. Assume $\varphi$ is not \emph{withdrawal-proof.} Then, there are $e=(R,\omega) \in \mathcal{E}^N,$ $\{i,j\} \subset N,$ and $(x_i,x_j) \in \mathbb{R}^2_+$ such that, if  $e'=(R_{N\setminus \{j\}}, \omega_{N \setminus \{j\}}),$ then 
\begin{equation}\label{realloc}
x_i+x_j=\varphi_i(e')+\omega_j,
\end{equation}
and 
\begin{equation}\label{pref}
x_kP_k\varphi_k(e) \text{ for each } k \in \{i,j\}.
\end{equation}
Assume $z(e)\geq 0.$ The case $z(e)\leq 0$ can be handled similarly.  By \eqref{pref}, $z(e)>0.$ By \emph{efficiency,} $\varphi_k(e) \leq p(R_k)$ for each $k \in N.$ By \eqref{pref}, $\varphi_k(e) < x_k$ for each $k \in \{i,j\}$ and therefore, by \eqref{realloc}, 
\begin{equation}\label{desigualdad1}
\varphi_i(e)+\varphi_j(e)<x_i+x_j=\varphi_i(e')+\omega_j. 
\end{equation}  
\textbf{Claim: there is $\boldsymbol{k^\star \in N\setminus \{i,j\}}$ such that $\boldsymbol{\varphi_{k^\star}(e') < \varphi_{k^\star}(e)}$.} Otherwise,
\begin{equation}\label{desigualdad2}
\sum_{k \in N\setminus \{i,j\}}\varphi_k(e')\geq \sum_{k \in N\setminus \{i,j\}}\varphi_k(e)
\end{equation}
and, since $\sum_{k \in N\setminus \{j\}}\omega_k=\sum_{k \in N\setminus \{j\}}\varphi_k(e'),$ by \eqref{desigualdad1} and \eqref{desigualdad2} we have $$\sum_{k \in N} \omega_k=\sum_{k \in N\setminus \{j\}}\varphi_k(e')+\omega_j>\sum_{k \in N}\varphi_k(e)=\sum_{k \in N} \omega_k,$$ which is absurd. This proves the Claim. 
 
 Now, by the Claim and \emph{efficiency}, $\varphi_{k^\star}(e') < \varphi_{k^\star}(e)\leq p(R_{k^\star}).$ This implies 
\begin{equation}\label{welf1}
\varphi_{k^\star}(e)P_{k^\star}\varphi_{k^\star}(e')
\end{equation} and also, by \emph{efficiency,} $z(e') \geq 0.$ By \eqref{pref},  $\varphi_j(e)\neq p(R_j)$ holds.  Then, by Lemma \ref{elb}, $\omega_j \leq p(R_j);$ and by the \emph{endowments lower bound,} $\varphi_j(e) \geq \omega_j.$ It follows from this and  \eqref{desigualdad1} that $0 \leq \varphi_j(e)-\omega_j<\varphi_i(e')-\varphi_i(e).$ Therefore, $\varphi_i(e')>\varphi_i(e).$ As $z(e') \geq 0,$ by \emph{efficiency,} $\varphi_i(e') \leq p(R_i).$  Thus, $\varphi_i(e) < \varphi_i(e') \leq p(R_i)$ and 
\begin{equation}\label{welf2}
\varphi_{i}(e')P_{i}\varphi_{i}(e).
\end{equation}  
Note that, as $z(e)\geq 0$ and $z(e')\geq 0,$ \eqref{welf1} and \eqref{welf2} contradict \emph{one-sided population  monotonicity}. We conclude that $\varphi$ is \emph{withdrawal-proof. } 
\end{proof}

As a consequence of the previous result and Theorem \ref{sequential}, the whole class of \emph{iterative} reallocation rules precludes this kind of manipulation. 

\begin{corollary}\label{cor} Each iterative reallocation rule is withdrawal-proof.
\end{corollary}

\subsection{Endowments-merging-proofness}\label{sectionmerging}

Another manipulation involving variable population is the following. Consider an economy and a pair of agents in that economy. One of those agents gives her endowment to the other and withdraws. The reallocation rule is applied  without the first agent and with the second agent's enlarged endowment. The allocation that the second agent obtains  could be divided between the two agents in such a way that each agent is at least as well off as she  would have been if the merging had not taken place, and at least one of them is better off. We require immunity to this sort of behavior: 

\vspace{5 pt}
\noindent
\textbf{Endowments-merging-proofness:} For each $e=(R,\omega) \in \mathcal{E}^N,$ each $\{i, j\} \subset N$ and each  $(x_i,x_j) \in \mathbb{R}^{2}_+$ such that $x_i+ x_j=\varphi_i(e'),$ where $e'= (R_{N\setminus\{j\}},\omega_i', \omega_{N\setminus\{i,j\}})$ and  $\omega_i'=\omega_i+\omega_j,$ it is not the case that $x_k R_k \varphi_k(e)$ for each $k \in \{i,j\},$ and  $x_kP_k\varphi_k(e)$ for at least one $k \in \{i,j\}.$
\vspace{5 pt}

Each rule in the class of iterative reallocation rules precludes such manipulations.

\begin{lemma}\label{Th emp}
Each iterative reallocation rule is  endowments-merging-proof.  
\end{lemma}
\begin{proof}
Let $\varphi$ be an iterative reallocation rule. By Theorem  \ref{sequential}, $\varphi$ is \emph{one-sided endowments monotonic}. By Lemma \ref{lemaEMimpliesEFF}, $\varphi$ is also \emph{efficient}.  Assume $\varphi$ is not \emph{endowments-merging-proof.} Then, there are $e=(R,\omega) \in \mathcal{E}^N,$ $\{i,j\} \subset N,$ and $(x_i,x_j) \in \mathbb{R}^2_+$ such that, if  $e'=(R_{N\setminus \{j\}}, \omega_i', \omega_{N \setminus \{i,j\}}),$ then 
\begin{equation}\label{realloc2}
x_i+x_j=\varphi_i(e'),
\end{equation}
\begin{equation}\label{pref21}
x_kR_k\varphi_k(e) \text{ for each } k \in \{i,j\},
\end{equation}
and 
\begin{equation}\label{pref2}
x_kP_k\varphi_k(e) \text{ for at least one } k \in \{i,j\}.
\end{equation}
Assume $z(e)\geq 0.$ By \eqref{pref2}, $z(e)>0.$ By \emph{efficiency,} $\varphi_k(e) \leq p(R_k)$ for each $k \in N.$ By \eqref{pref21} and \eqref{pref2}, $x_k \geq \varphi_k(e)$ for each $k \in \{i,j\}$ and $x_k > \varphi_k(e)$ for at least one $k \in \{i,j\}.$ Therefore, by \eqref{realloc2}, 
\begin{equation}\label{desigualdad3}
\varphi_i(e')=x_i+x_k>\varphi_i(e)+\varphi_j(e). 
\end{equation}  
\textbf{Claim 1: there is $\boldsymbol{k^\star \in N\setminus \{i,j\}}$ such that $\boldsymbol{\varphi_{k^\star}(e') < \varphi_{k^\star}(e)}$.} Otherwise,
\begin{equation}\label{desigualdad4}
\sum_{k \in N\setminus \{i,j\}}\varphi_k(e')\geq \sum_{k \in N\setminus \{i,j\}}\varphi_k(e)
\end{equation}
and, since $\sum_{k \in N\setminus \{j\}}\omega_k=\sum_{k \in N\setminus \{j\}}\varphi_k(e'),$ by \eqref{desigualdad3} and \eqref{desigualdad4} we have $$\sum_{k \in N} \omega_k=\sum_{k \in N\setminus \{j\}}\varphi_k(e')>\sum_{k \in N}\varphi_k(e)=\sum_{k \in N} \omega_k,$$ which is absurd. 
 
 Now, by Claim 1 and \emph{efficiency}, $\varphi_{k^\star}(e') < \varphi_{k^\star}(e)\leq p(R_{k^\star}),$ which implies  $z(e') \geq 0.$ Let $e''=(R_{N\setminus\{j\}}, \omega_{N\setminus\{j\}}).$ As $z(e') \geq 0,$ it follows that $z(e'')\geq 0.$ By Corollary \ref{cor}, $\varphi$ is \emph{withdrawal-proof,} which implies that 
\begin{equation}\label{wp1}
\varphi_i(e)+\varphi_j(e)\geq \varphi_i(e'')+\omega_j.
\end{equation}
Since $(\omega_i', \omega_{N\setminus\{i,j\}}) \geqslant \omega_{N\setminus\{j\}},$ by \emph{one-sided endowments monotonicity,} $\varphi_k(e')R_k\varphi_k(e'')$ for each $k \in N\setminus\{j\}.$ Then \emph{efficiency} implies  
\begin{equation}\label{emon}
\varphi_k(e')\geq \varphi_k(e'') \text{ for each } k \in N\setminus \{j\}. 
\end{equation}
Combining \eqref{desigualdad3} and \eqref{wp1} we obtain 
\begin{equation}\label{desigualdad5}
\varphi_i(e')>\varphi_i(e'')+\omega_j.
\end{equation}
\textbf{Claim 2: there is $\boldsymbol{k^{\star\star} \in N\setminus \{i,j\}}$ such that $\boldsymbol{\varphi_{k^{\star\star}}(e'') > \varphi_{k^{\star\star}}(e')}$.} Otherwise,
\begin{equation}\label{desigualdad6}
\sum_{k \in N\setminus \{i,j\}}\varphi_k(e'')\leq \sum_{k \in N\setminus \{i,j\}}\varphi_k(e')
\end{equation}
and, since $\sum_{k \in N\setminus \{j\}}\omega_k=\sum_{k \in N\setminus \{j\}}\varphi_k(e''),$ by \eqref{desigualdad5} and \eqref{desigualdad6} we have $$\sum_{k \in N} \omega_k=\omega_j+\sum_{k \in N\setminus \{j\}}\varphi_k(e'')<\sum_{k \in N\setminus\{j\}}\varphi_k(e')=\sum_{k \in N} \omega_k,$$ which is absurd.  

Therefore, by Claim 2, there is $k^{\star \star} \in N\setminus \{i,j\}$ such that $\varphi_{k^{\star \star}}(e'')>\varphi_{k^{\star \star}}(e').$ This contradicts \eqref{emon}. We conclude that $\varphi$ is \emph{endowments-merging-proof.} 
\end{proof}

\subsection{Endowments-splitting-proofness}\label{splitting}

Consider an economy and assume that an agent in the economy transfers some of her endowment to another agent that was not initially present; the rule is applied, and the guest transfers her assignment to the agent who invited her in. The first agent could obtain an amount that she prefers to her initial assignment. We require immunity to this type of behavior:

\vspace{5 pt}
\noindent
\textbf{Endowments-splitting-proofness:}  
For each $e=(R,\omega) \in
\mathcal{E}^N,$ each $i \in  N,$ each $j \notin N,$  each  $R_j \in \mathcal{R},$ and each $(\omega_i', \omega_j') \in \mathbb{R}^{2}_+$ such that 
$\omega_i'+ \omega_j'=\omega_i,$ we have $\varphi_i(e) R_i [\varphi_i(e')+\varphi_j(e')],$ with $e' = (R, R_{j},\omega_i', \omega_{N\setminus\{i\}}, \omega_j').$
\vspace{5 pt}

Not all iterative reallocation rules satisfy this property. The following example shows that 
the uniform reallocation rule violates \emph{endowments-splitting-proofness.} 

\begin{example}
\emph{ Let $e=(R,\omega) \in \mathcal{E}^{\{1,2,3\}}$ be such that $p(R_1)=4,$ $p(R_2)=0,$ $p(R_3)=\omega_1=\omega_2=2$ and $\omega_3=1.$ Then,  $u_1(e)=3,$ $u_2(e)=0,$ and $u_3(e)=2.$  Next, let $R_4 \in \mathcal{R}$ be such that $p(R_4)=4$ and  let $\omega_4'=1.$ Consider the economy $e'=(R,R_4, \omega_1', \omega_{\{2,3\}}, \omega_4') \in \mathcal{E}^{\{1,2,3,4\}}$ with $\omega'_1=1$ (notice that $\omega_1=\omega_4'+\omega_1'$). It follows that $u_1(e')=u_3(e')=u_4(e')=\frac{5}{3},$ $u_2(e')=0,$ and   $u_1(e')+u_4(e')=\frac{10}{3} \ P_1 \  3=u_1(e).$  This implies that $u$ is not  \emph{endowments-splitting-proof.}
}  \hfill $\Diamond$
\end{example}

\noindent Priority reallocation rules\footnote{Given a linear $\prec$ order over the set of potential agents $\mathbb{N}$, the \emph{priority reallocation rule} $\varphi^\prec$ for economies with excess demand (supply) satiates all suppliers (demanders) and demanders (suppliers) according to order $\prec$, respecting the endowments lower bound.  For economies with excess supply, a symmetric procedure is performed. See Appendix \ref{appendix} for a formal definition. It is easy to see that such reallocation rules are iterative in our sense.} violate the property as well. 

\begin{example}
\emph{Consider $\prec$ as the usual ``less than'' order in $\mathbb{N}$. Let $e=(R, \omega) \in \mathcal{E}^{\{1,3,4\}}$ be such that $p(R_1)=0,$ $\omega_1=4,$ $p(R_3)=p(R_4)=6,$ and  $\omega_3=\omega_4=2.$ Then, $\varphi^\prec_1(e)=0,$ $\varphi^\prec_3(e)=6,$ and $\varphi^\prec_4(e)=2.$ Next, let $R_2 \in \mathcal{R}$ be such that $p(R_2)=4$ and let $\omega_2=1.$ Consider the economy $e'=(R,R_2, \omega_{\{1,3\}},  \omega_2, \omega_4') \in \mathcal{E}^{\{1,2,3,4\}}$ with $\omega_4'=1$ (notice that $\omega_4=\omega_4'+\omega_2=2$). It follows that $\varphi^\prec_1(e')=0,$ $\varphi^\prec_2(e')=4,$  $\varphi^\prec_3(e')=3,$ and  $\varphi^\prec_4(e')=1.$ However, $\varphi^\prec_2(e')+\varphi^\prec_4(e')=4+1=5 \ P_4 \  2=\varphi^\prec_4(e).$  This implies that $\varphi^\prec$ is not  \emph{endowments-splitting-proof.}
}  \hfill $\Diamond$
\end{example}

\noindent However, the proportional reallocation rule is immune to endowments' splitting:

\vspace{5 pt}

\begin{remark}\label{propprop}
The proportional reallocation rule is endowments-splitting-proof.
\end{remark}
\begin{proof}
Suppose $\varphi^p$ is not \emph{endowments-splitting-proof.} Then, there are $e=(R,\omega) \in \mathcal{E}^N,$ $i \in  N,$ $j \notin N,$ $R_j \in \mathcal{R},$ and $(\omega_i', \omega_j') \in \mathbb{R}^{2}_+$ with 
$\omega_i'+ \omega_j'=\omega_i$ such that, if  $e' = (R, R_{j},\omega_i', \omega_{N\setminus\{i\}}, \omega_j'),$ then 
\begin{equation}\label{propo1}
[\varphi_i^p(e')+\varphi_j^p(e')] P_i \varphi_i^p(e).
\end{equation}
Consider first the case $z(e)\geq0.$ By \eqref{propo1}, $\varphi_i^p(e)<p(R_i)$ and therefore 
\begin{equation}\label{propo2}
\lambda(e)\omega_i=\varphi_i^p(e) < \varphi_i^p(e')+\varphi_j^p(e').
\end{equation} 
Since $z(e') \geq z(e)\geq 0,$ 
\begin{equation}\label{propo3}
\varphi_i^p(e')+\varphi_j^p(e')\leq \lambda(e')\omega_i'+\lambda(e')\omega_j'=\lambda(e')\omega_i.
\end{equation}
By \eqref{propo2} and \eqref{propo3},
\begin{equation}\label{propo4}
\lambda(e)<\lambda(e').
\end{equation}
It follows that there is $k \in N \setminus \{i\}$ such that $\varphi^p_k(e')<\varphi^p_k(e).$ Otherwise, there is a violation of feasibility by \eqref{propo2}. Then, $\lambda(e')\omega_k=\varphi^p_k(e')<\varphi^p_k(e)\leq\lambda(e)\omega_k,$ which implies $\lambda(e')<\lambda(e),$ contradicting \eqref{propo4}. 
If   $z(e) \leq 0$ and $z(e') \leq 0,$ the proof is similar to the previous one. Assume then that $z(e) \leq 0$ and $z(e') \geq 0.$ By \eqref{propo1}, $\varphi_i^p(e)>p(R_i)\geq \varphi_i^p(e').$ This implies the existence of $k \in N \setminus \{i\}$ such that $\varphi_k^p(e) < \varphi_k^p(e').$ But then, $$p(R_k) \leq \varphi_k^p(e) < \varphi_k^p(e') \leq p(R_k),$$ which is a contradiction. 
Therefore, $\varphi^p$ is \emph{endowments-splitting-proof.}      
\end{proof}

\subsection{Pre-delivery-proofness}\label{sectionpredelivery}

Consider now the case in which  one agent makes  a ``pre-delivery'' to some other agent   of the trade that  this second agent would be assigned if she had participated with everyone else. After the rule is applied, the first agent  may end up with an amount she prefers to his assignment  if she had not carried out the pre-delivery. We require immunity to this sort of behavior.  

\vspace{5 pt}
\noindent
\textbf{Pre-delivery-proofness:}  
For each $e=(R,\omega) \in
\mathcal{E}^N$ and each $\{i, j\} \subset N$ such that $\omega_i + \omega_j - \varphi_j(e) \geq 0, $ $\varphi_i(e) R_i \varphi_i(e')$ where $e'= (R_{N\setminus\{j\}},\omega_i', \omega_{N\setminus\{i,j\}})$ and $\omega_i'=\omega_i + \omega_j- \varphi_j(e).$
\vspace{5 pt}

 The endowments rule, that in each economy assigns to  agents their own endowment, is trivially \emph{pre-delivery-proof}. However, as long as we require the rule to be \emph{efficient} and  \emph{own-peak-only}, and  meet the \emph{endowments lower bound},  we reach an impossibility.

\begin{theorem}\label{proppre} No efficient and own-peak-only reallocation rule that meets the endowments lower bound is pre-delivery-proof. 
\end{theorem}
\begin{proof}
Let reallocation rule $\varphi$ be  efficient, own-peak-only and  meet the endowments lower bound. Let $e=(R,\omega) \in \mathcal{E}^{\{1,2,3\}}$ be such that $0<p(R_1)=\omega_2=\omega_3<\omega_1<p(R_2)=p(R_3).$ Then $z(e)> 0$ and, as $p(R_1)<\omega_1,$ by Lemma \ref{elb} we have  $\varphi_1(e)= p(R_1).$
By feasibility, there is $i^\star \in \{2,3\}$ such that $\varphi_{i^\star}(e)<\omega_1.$ Assume, without loss of generality, that $i^\star=2.$ Let $\omega_{2}'=\omega_{2}+\omega_1-\varphi_1(e).$ Then $\omega_{2}'=\omega_1.$ Consider now the economy $e'=(R_{\{2,3\}}, \omega_2', \omega_3).$ It follows that $z(e')>0.$ By \emph{efficiency,} $\varphi_2(e') \leq p(R_2),$ and since $p(R_2)>\omega_2',$ by the \emph{endowments lower bound} we have $\varphi_2(e') \geq \omega_2'=\omega_1.$ By feasibility then, $\varphi_2(e')=\omega_1.$ Therefore, $\varphi_2(e)<\omega_1=\varphi_2(e')<p(R_2),$ which implies $\varphi_2(e')P_2\varphi_2(e).$ Thus, $\varphi$ is not \emph{pre-delivery-proof.}  
\end{proof}

Of course, the previous result extends to the whole class of iterative reallocation rules.  
\begin{corollary}\label{corpre} No iterative reallocation rule is pre-delivery-proof.
\end{corollary}

\section{Final comments}\label{sectionfinal}

We conclude with some remarks. One may ask whether the definition of \emph{withdrawal-proofness}  can be made with just one of the agents involved in the manipulation strictly improving. However, not even the uniform reallocation rule  satisfies this variant, as the following example shows:

\begin{example} \rm
Consider $e=(R,\omega)^{\{1,2,3,4\}}$ with $p(R_1)=p(R_4)=1,$ $p(R_2)=4,$ $p(R_3)=3,$ and $\omega_1=\omega_4=3,$ $\omega_2=\omega_3=1.$ Then, $z(e)>0$ and $u_1(e)=u_4(e)=1,$ $u_2(e)=u_3(e)=3.$  When agent 4 withdraws, if $e'=(R_{\{1,2,3\}},\omega_{\{1,2,3\}}),$ then $u_1(e')=1,$ $u_2(e')=u_3(e')=2.$  So $x_2=4 \ P_2 \ 3=u_2(e),$ $x_4=1 \ R_1 \ 1=u_4(e),$ and $x_2+x_4=5=u_2(e')+\omega_4.$ 
\hfill $\Diamond$
\end{example}

\noindent In \cite{thomson2014new}, the property of \emph{withdrawal-proofness} is presented in this variant: one of the agents involved in the manipulation can be indifferent between the amount she gets from the rule and the amount she gets after the manipulation is performed.  However, in all the impossibility examples presented there both agents get strictly better off, so our version does not hold in those examples (classical multi-commodity  exchange model with homothetic and quasi-linear preferences) either.   

The results obtained in our paper are in sharp contrast with the findings 
in models with several goods and classical preferences. 
The Walrasian reallocation rule is neither  \emph{withdrawing-proof,} nor \emph{endowments-merging-proof,} nor \emph{endowments-splitting-proof}. These negative results are obtained by \cite{thomson2014new} in two classical subdomains: (i) the domain of economies in which preferences are homothetic and strictly convex, and individual endowments are proportional, and (ii) the domain of economies in which preferences are quasi-linear and strictly convex. The Walrasian reallocation rule, however, is \emph{pre-delivery-proof} on the classical domain \citep[see][]{thomson2014new}. ``Constrained dictatorial rules'', defined by maximizing the welfare of a particular agent subject to each of the others finding their assignment at least as desirable as their endowment, satisfy none of these various requirements either \citep[see][]{thomsonconstrained22}.

Several questions are left open for further research. For example, finding  out which properties are required for a rule to be \emph{endowments-merging-proof} or \emph{endowments-splitting-proof}, and knowing whether iterative reallocation rules exhaust the  rules that are \emph{withdrawal-proof} or \emph{endowments-merging-proof}.  Coalitional versions of the manipulations studied in Section \ref{sectionvariable} seem worth studying as well.

\appendix

\section{Independence of axioms in Theorem  \ref{sequential}}\label{appendix}

In order to study the independence of axioms in the characterization of Theorem \ref{sequential}, next we consider several reallocation rules. For each $N \in \mathcal{N}$ and each $e\in \mathcal{E}^N$, let $N^+(e)=\{i \in N : p(R_i)>\omega_i\}$ be the set of \emph{demanders} of $e$. Agents in $N\setminus N^+$ are called \emph{suppliers}. Let $\mathcal{S}(e)\equiv \sum_{j \in N\setminus N^+}(\omega_j-p(R_j))$.

Given a linear $\preceq$ order over the set of potential agents $\mathbb{N}$, the \emph{priority reallocation rule} $\varphi^\preceq$ for each economy with excess demand (supply) satiates all suppliers (demanders) and demanders (suppliers) according to order $\preceq$, respecting the endowments lower bound. So if $e=(R,\omega)\in \mathcal{E}^N$ is such that $z(e)\geq 0,$   $$\varphi^\preceq_i(e)=\left\{\begin{array}{l l }
p(R_i) & \text{if } i\in N\setminus N^+(e)\\
\min\left\{p(R_i), \omega_i+\mathcal{S}(e)-\sum_{j \in N^+ : j \prec i}\Delta \varphi_j(e)\right\} & \text{if } i\in N^+(e)\\
\end{array}\right.$$
In case $z(e)<0$, the rule is defined similarly. Priority reallocation rules allow the definition of a reallocation rule which is not \emph{one-sided population monotonic}. 

\vspace{5 pt}

\noindent \textbf{Reallocation rule  $\boldsymbol{\overline{\varphi}}$:} for each $N \in \mathcal{N},$ each $e\in \mathcal{E}^N,$ and each $i \in N,$
$$\overline{\varphi}_i(e)=\left\{\begin{array}{l l }
\varphi^\preceq_i(e) & \text{if } |N| \text{ is odd}\\
\varphi^\succeq_i(e) & \text{if } |N| \text{ is even}\\
\end{array}\right.$$
where $\succeq$ is the dual of $\preceq$.
\vspace{5 pt}

Next, let us recall  the celebrated \emph{uniform} rule, first characterized by \cite{Sprumont91}. 

\vspace{5 pt}

\noindent \textbf{Uniform  rule, $\boldsymbol{\varphi^u}$:} for each $N \in \mathcal{N},$ each $e\in \mathcal{E}^N,$ and each $i \in N,$
$$\varphi^u_i(e)=\left\{\begin{array}{l l }
\min\{p(R_i), \lambda(e)\} & \text{if } z(e)\geq 0\\
\max\{p(R_i), \lambda(e)\} & \text{if } z(e)< 0\\
\end{array}\right.$$
where $\lambda(e)$ and solves $\sum_{j \in N}\varphi^u_j(e)=\sum_{j\in N}\omega_j.$

\vspace{5 pt}

\noindent Since this rule does not take into account individual endowments, it trivially does not meet the \emph{endowments lower bound}.

The following rule satiates as many agents as possible. For economies with excess demand, demanders are satiated according to their claims. First minimal demands are satiated uniformly. If there is some supply left, then the next smallest demands are satiated, and so on. This reallocation rule is not \emph{one-sided endowments monotonic}.

\vspace{5 pt}

\noindent \textbf{Maximally satiating reallocation rule $\boldsymbol{\varphi^{\text{max}}}$:}  For each $N \in \mathcal{N}$ and each $e\in \mathcal{E}^N$ such that $z(e)\geq 0,$ partition  $N^+(e)$ into subsets $N_1,N_2, \ldots, N_s$ such that (i) for each $t \in \{1, \ldots, s\}$, $p(R_i)-\omega_i=p(R_j)-\omega_j$ for each $i,j \in N_t$, and (ii) $p(R_i)-\omega_i<p(R_j)-\omega_j$ if $i \in N_r$, $j \in N_s$, and $r<s.$ 
Then,  $$\varphi^{\text{max}}_i(e)=p(R_i)$$ if 
 $i\in N\setminus N^+$, and 
$$\varphi^{\text{max}}_i(e)=\min\left\{p(R_i), \omega_i+\frac{1}{|N_t|}\left(\mathcal{S}(e)-\sum_{j \in \cup_{r=1}^{t-1} N_r}\Delta \varphi_j(e)\right)\right\}$$ if $i \in N_t$ and $t \in \{1, \ldots, s\}.$ The formula when $z(e)<0$ is obtained similarly. 

\vspace{5 pt}

\begin{table}[h] 
\small
\centering 
\begin{threeparttable}
\begin{tabular}{|c|c|c|c|c|}
\hline
  & Own-peak-only & Endow LB & OS endow mon & OS pop mon\\
 \hline \hline
$\overline{\varphi}$  & $+$ & $+$ & $+$ & $-$ \\
\hline \hline
$\varphi^{\text{max}}$ & $+$ & $+$ & $-$ & $+$ \\
\hline \hline

$\varphi^u$ & $+$ & $-$ & $+$ & $+$ \\
\hline \hline
? & $-$& $+$& $+$& $+$\\
\hline
\end{tabular}
\end{threeparttable}
\caption{\emph{Independence of axioms in the characterization of Theorem  \ref{sequential}. }}\label{tabla caracterizaciones}
\end{table}

Each one of the previously presented rules satisfies all properties of the characterization in Theorem  \ref{sequential} except one. This is shown in Table \ref{tabla caracterizaciones}. For example, reallocation rule $\varphi^u$ satisfies both monotonicity properties and the \emph{own-peak-only} property, but does not meet the \emph{endowments lower bound}. Whether there is a reallocation rule that satisfies all properties except the \emph{own-peak-only} property remains an open question.

\section{Iterative reallocation rules are neither peak-only nor strategy-proof}\label{segundo}

The following strong form of the \emph{own-peak-only} property requires that if two profiles of preferences have the same peaks, then the reallocations recommended by the rule in both profiles are the same.

\vspace{5 pt}

\noindent \textbf{Peak-only:} For each $e=(R, \omega) \in \mathcal{E}^N$ and each $ R' \in  \mathcal{R}$ such that $p(R_i')=p(R_i)$ for each $i \in N$, if $e'=(R', \omega)$ then $\varphi(e')=\varphi(e).$

\vspace{5 pt}

\noindent Iterative reallocation rules are not in general \emph{peak-only}, because they can be \emph{bossy}: it could be that a change in one agent's preferences (preserving the peak amount) only affects other agents' reallocations (see Example \ref{ni PO ni SP}). When a rule is \emph{own-peak-only} and \emph{non-bossy}\footnote{\textbf{Non-bossiness:} For each $e=(R, \omega) \in \mathcal{E}^N$, each $i \in N$,  and each $ R'_i \in  \mathcal{R}$,  if $e'=(R'_i, R_{-i}, \omega)$ and $\varphi_i(e')=\varphi_i(e')$ then $\varphi(e')=\varphi(e).$}, then it is also \emph{peak-only}. Sequential allocation rules in \cite{Barbera97} satisfy \emph{replacement monotonicity}, which in turn implies \emph{non-bossiness}. For this reason, \emph{own-peak-onliness} is equivalent to \emph{peak-onliness} for those rules. 

A reallocation  rule is \emph{strategy-proof} if, for each agent, truth-telling is always optimal, regardless of the
preferences declared by the other agents. Formally,

\vspace{5 pt}

\noindent \textbf{Strategy-proofness:} For each $e=(R, \omega) \in \mathcal{E}^N$, each $i \in N$, and each $R_i' \in \mathcal{R}$, if $e'=(R'_i, R_{-i}, \omega)$ then $\varphi_i(e)R_i\varphi_i(e').$

\vspace{5 pt}

\noindent \emph{Strategy-proofness} is a strong property. Together with \emph{efficiency}, it implies the \emph{own-peak-only} property for reallocation rules  \citep[this is proven, for example, in Lemma 3 of][]{Klaus98}. Since we do not impose \emph{strategy-proofness} in our paper, we need to explicitly invoke the \emph{own-peak-only} property. 

To see that iterative reallocation rules need not be \emph{peak-only} nor \emph{strategy-proof}, consider  reallocation rule $\varphi^\star$  that works as follows.\footnote{Here we use the terminology presented in  Appendix \ref{appendix}.} For each $N \in \mathcal{N}$ and each $(R, \omega) \in \mathcal{E}^N$ with $z(e)\geq 0$, let $i^\star$ be the agent in $N \setminus N^+(e)$ (the set of suppliers) with the lowest index.\footnote{If $N \setminus N^+(e)=\emptyset$, simply allocate to each agent their endowment.} Next, define a linear  order $\ll$  over $N^+(e)$ (the set of demanders)  saying that  
 $i \ll j$ if and only if either $p(R_i)-\omega_i<p(R_j)-\omega_j$ or $p(R_i)-\omega_i=p(R_j)-\omega_j$ and $i<j$. If $0P_{i^\star}\sum_{j \in N}\omega_j$, then we apply the priority reallocation  rule according to  $\ll$, i.e., $\varphi^\star(e)=\varphi^{\ll}(e)$.  If $\sum_{j \in N}\omega_jR_{i^\star}0$, then we apply the uniform reallocation rule, i.e., $\varphi^\star(e)=\varphi^u(e)$. The definition of the rule for economies with excess supply is similar. It is easy to see that reallocation rule $\varphi^\star$ belongs to the class of iterative reallocation rules. The following example shows that $\varphi^\star$ is neither \emph{peak-only} nor \emph{strategy-proof}.

\begin{example}\label{ni PO ni SP} \rm
Let $e=(R, \omega) \in \mathcal{E}^{\{1,2,3\}}$ be such that $\omega_1=9$, $\omega_2=1$, $\omega_3=4$, $p(R_1)=1$ and $0 \ P_1 \ 14$, $p(R_2)=7$, and $p(R_3)=9$. Since $p(R_3)-\omega_3=9-4=5<6=7-1=p(R_2)-\omega_2$, we have $\varphi^\star_1(e)=1$, $\varphi^\star_2(e)=4$, and $\varphi^\star_3(e)=9$. Let $\widetilde{R}_1 \in \mathcal{R}$ be such that $p(\widetilde{R}_1)=p(R_1)$ and $14 \ \widetilde{P}_1 \ 0$, an consider economy $\widetilde{e}=(\widetilde{R}_1, R_{-1}, \omega)$. Then, $\varphi_1(\widetilde{e})=1$, $\varphi_2(\widetilde{e})=5$, and $\varphi_3(\widetilde{e})=8$. Thus, although $\varphi^\star$ is \emph{own-peak-only}, it is certainly not \emph{peak-only}. 

To see that $\varphi^\star$ does not satisfy \emph{strategy-proofness}, let $R_2' \in \mathcal{R}$ be such that $p(R_2')=5.5$ and consider economy $e'=(R_2', R_{-2}, \omega)$. Since $p(R_2')-\omega_2=4.5<5=9-4=p(R_3)-\omega_3$, we have $\varphi^\star_2(e')=5.5$. Therefore, $\varphi^\star_2(e')=5.5 \ P_2 \ 4=\varphi_2^\star(e).$ \hfill $\Diamond$
\end{example} 

Finally, it is worth mentioning that the weakly sequential rules in \cite{Bonifacio15} are immune to misrepresentation of preferences and endowments in a very general way: they are \emph{bribe-proof}.\footnote{This concept was introduced by \cite{schummer2000manipulation} and applied to the model with a social endowment by \cite{masso2007bribe}.}   This property does not allow any group of agents to compensate one of its subgroups to misrepresent their preferences or endowments in order that, after an appropriate redistribution of what the rule reallocates to the group (adjusted by the resource surplus or deficit they all engage in by misreporting), (i) each agent in the misrepresenting subgroup obtain a preferred amount, and (ii) the rest of the agents in the group are not made worse-off.


\bibliographystyle{ecta}
\bibliography{library}

\end{document}